\documentclass[
    aps,
    prx,
    twocolumn,
    superscriptaddress,
    nofootinbib,
    floatfix,
    citeautoscript,
    10pt
]{revtex4-2}

\usepackage[linesnumbered,ruled]{algorithm2e}
\SetArgSty{textnormal}

\usepackage{amsmath}
\usepackage{amssymb}
\usepackage{amsthm}
\usepackage{bm} % bold math
\usepackage{bbm} % black board math
\usepackage{braket}
\usepackage{dsfont}
\usepackage[colorlinks]{hyperref}
\hypersetup{
    colorlinks  = true,
    citecolor   = PineGreen,
    linkcolor   = MidnightBlue,
    urlcolor    = PineGreen
}
\usepackage{xurl}
\hypersetup{breaklinks=true}
\usepackage[bb=stix]{mathalpha} % blackboard math
\usepackage{mathtools}
\usepackage{thmtools}
\usepackage[dvipsnames]{xcolor}

\usepackage{etoolbox}
\apptocmd{\sloppy}{\hbadness 10000\relax}{}{}

\usepackage{silence}
        \DeclarePairedDelimiter{\abs}{\lvert}{\rvert} % absolute value
        \DeclareMathOperator{\ancilla}{a} % ancilla register
        \DeclareMathOperator{\data}{d} % data register
        \DeclareMathOperator{\permutation}{p} % permutation register
        \DeclareMathOperator{\id}{id} % identity (permutation)
        \DeclareMathOperator{\bin}{bin} % binary

        \newcommand*{\control}{\texttt{C}} % control
        \newcommand*{\h}{\texttt{H}} % Hadamard gate
        \newcommand*{\ch}{\texttt{CH}} % controlled Hadamard gate
        \newcommand*{\x}{\texttt{X}} % Pauli-X gate
        \newcommand*{\cx}{\texttt{CX}} % controlled Pauli-X gate
        \newcommand*{\ry}{\texttt{RY}} % rotational Y-gate
        \newcommand*{\rz}{\texttt{RZ}} % rotational Z-gate
        \newcommand*{\cry}{\texttt{CRY}} % controlled rotational Y-gate
        \newcommand*{\swap}{\texttt{SWAP}} % SWAP gate
        
        \newcommand*{\N}{\mathbb{N}}

        \newcommand*{\C}{\mathbb{C}}
        \newcommand*{\hil}{\mathcal{H}}
        \newcommand*{\bigo}{\mathcal{O}} % big O notation
        \newcommand*{\hamming}{\Delta} % Hamming distance
        \newcommand*{\numbits}[1]{\abs{{#1}}_{\bin}} % length of binary representation
        \newcommand*{\scriptin}{\raisebox{0.15ex}{$\scriptscriptstyle\in$}} % "element of" symbol suited for sub- and superscripts
            \newcommand*{\symgroup}{S}

            \newcommand*{\qubitc}{\ensuremath{\mathbf{Q}}} % qubit count
            \newcommand*{\gatec}{\ensuremath{\mathbf{G}}} % gate count
            \newcommand*{\cyclec}{\ensuremath{\mathbf{C}}} % cycle count

\declaretheorem[style=plain]{theorem}

\declaretheorem[style=plain,sibling=theorem]{corollary}
\declaretheorem[style=plain,sibling=theorem]{lemma}

\begin{document}

\title{Quantum Fisher-Yates shuffle: Unifying methods for generating uniform superpositions of permutations}

\author{Lennart \surname{Binkowski}}
\email{lennart.binkowski@itp.uni-hannover.de}
\affiliation{Institut für Theoretische Physik, Leibniz Universität Hannover}
\author{Marvin \surname{Schwiering}}
\affiliation{Institut für Theoretische Physik, Leibniz Universität Hannover}

\begin{abstract}
    Uniform superpositions over permutations play a central role in quantum error correction, cryptography, and combinatorial optimisation.
    We introduce a simple yet powerful quantisation of the classical Fisher–Yates shuffle, yielding a suite of efficient quantum algorithms for preparing such superpositions on composite registers.
    Our method replaces classical randomness with coherent control, enabling five variants that differ in their output structure and entanglement with ancillary systems.
    We demonstrate that this construction achieves the best known combination of asymptotic resources among all existing approaches, requiring only $\bigo(n \log(n))$ qubits and $\bigo(n^{2} \log(n))$ gates and circuit depth.
    These results position the quantum Fisher–Yates shuffle as a strong candidate for optimality within this class of algorithms.
    Our work unifies several prior constructions under a single, transparent framework and opens up new directions for quantum state preparation using classical combinatorial insights.
    Our implementation in Qiskit is available as open-source code, supporting reproducibility and future exploration of quantum permutation-based algorithms. 
\end{abstract}

\maketitle

\section{\label{section:Introduction}Introduction}

Uniform superpositions of permutations have emerged as a powerful and recurring motif across multiple domains of quantum computing, underpinning advances in quantum error correction, quantum cryptography, and quantum-enhanced optimisation.
Their utility arises from the fundamental role that permutation symmetry plays in encoding, security, and search processes on quantum devices.

In quantum error correction, permutation-invariant codes exploit the symmetry under qubit permutations to protect against collective noise, spontaneous decay, and qubit loss~\cite{Barenco1997StabilizationOfQuantumComputationsBySymmetrization,Pollatsek2004PermutationallyInvariantCodesForQuantumErrorCorrection,Ouyang2014PermutationInvariantQuantumCodes,Ouyang2016PermutationInvariantCodesEncodingMoreThanOneQubit,Wu2019InitializingAPermutationInvariantQuantumErrorCorrectionCode}.
These codes enable robust encoding schemes that remain invariant under subsystem exchange, facilitating resilience in noisy intermediate-scale quantum (NISQ) environments.

In quantum cryptography, superpositions over permutations enable provably secure encryption protocols, such as the quantum permutation pad~\cite{Kuang2022QuantumPermutationPadForUniversalQuantumSafeCryptography}, and support keyless secure communication~\cite{Yang2013QuantumNoKeyProtocolForSecureCommunicationOfClassicalMessage}.
These structures also play an increasingly prominent role in post-quantum cryptography~\cite{Alagic2024PostQuantumSecurityOfTweakableEvenMansourAndApplications,Budroni2024SokMethodsForSamplingRandomPermutationsInPostQuantumCryptography,Majenz2024PermutationSuperpositionOraclesForQuantumQueryLowerBounds}, where they serve as building blocks for quantum-secure primitives and as analytical tools in oracle-based lower bounds.

In quantum optimisation, many combinatorial problems—most notably the Travelling Salesperson Problem (TSP)—are naturally encoded as search problems over the space of permutations.
Uniform superpositions of permutations are thus integral to quantum heuristic methods such as Grover-based search~\cite{Grover1996AFastQuantumMechanicalAlgorithmForDatabaseSearch,Zhu2022ARealizableGASBasedQuantumAlgorithmForTravelingSalesmanProblem,Bang2012AQuantumHeuristicAlgorithmForTheTravelingSalesmanProblem,Sato2023EmbeddingAllFeasibleSolutionsOfTravelingSalesmanProblemByDivideAndConquerQuantumSearch,Bai2025AQuantumSpeedupAlgorithmForTSPBasedOnQuantumDynamicProgrammingWithVeryFewQubits} and QAOA~\cite{Farhi2014AQuantumApproximateOptimizationAlgorithm,Hadfield2019FromTheQuantumApproximateOptimizationAlgorithmToAQuantumAlternatingOperatorAnsatz,Baertschi2020GroverMixersForQAOAShiftingComplexityFromMixerDesignToStatePreparation,Bourreau2023IndirectQuantumApproximateOptimizationAlgorithmsApplicationsToTheTSP}, where they facilitate exploration of the feasible solution space.
These approaches extend naturally to more complex scheduling and routing problems~\cite{Amaro2022ACaseStudyOfVariationalQuantumAlgorithmsForAJobShopSchedulingProblem,Palackal2023GraphControlledPermutationMixersInQAOAForTheFlexibleJobShopProblem,Kossmann2023OpenShopSchedulingWithHardConstraints,Kurowski2023ApplicationOfQuantumApproximateOptimizationAlgorithmToJobShopSchedulingProblem,Bogatyrev2023ApplicationOfGroversAlgorithmInRouteOptimization,Azad2023SolvingVehicleRoutingProblemUsingQuantumApproximateOptimizationAlgorithm,Palackal2023QuantumAssistedSolutionPathsForTheCapacitatedVehicleRoutingProblem}, where permutations define the structure of feasible assignments.

Beyond these applications, uniform superpositions over permutations also appear in theoretical contexts such as quantum complexity theory~\cite{Aharonov2003AdiabaticQuantumStateGenerationAndStatisticalZeroKnowledge} and entanglement theory~\cite{Toth2009EntanglementAndPermutationalSymmetry}, further underscoring their foundational character.

The first known quantum circuit for generating a uniform superposition over all permutations was proposed by~\citet{Barenco1997StabilizationOfQuantumComputationsBySymmetrization}, employing \emph{monotone factorisation} to decompose permutations into a sequence of transpositions, each conditioned via ancilla registers prepared in uniform superpositions. Their circuit exemplifies early instances of the now-standard Linear Combination of Unitaries (LCU) framework~\cite{Childs2012HamiltonianSimulationUsingLinearCombinationsOfUnitaryOperations}, with each controlled \swap-gate implementing a unitary in the sum.

Subsequent constructions have explored alternative permutation encodings. \citet{Chiew2019GraphComparisonViaNonlinearQuantumSearch} employed \emph{mixed radix} representations~\cite{Sedgewick1977PermutationGenerationMethods}, encoding permutations into integer vectors and reconstructing them via sequences of adjacent transpositions. Similar encodings, including Lehmer codes~\cite{Lehmer1960TeachingCombinatorialTricksToAComputer}, were used in~\cite{Marsh2020CombinatorialOptimizationViaHighlyEfficientQuantumWalks}.
Meanwhile, \citet{Baertschi2020GroverMixersForQAOAShiftingComplexityFromMixerDesignToStatePreparation} proposed a method for generating uniform superpositions of $n \times n$ permutation matrices by iteratively building row vectors under column constraints, using ancilla registers to enforce orthogonality conditions in superposition.

More recently, \citet{Majenz2024PermutationSuperpositionOraclesForQuantumQueryLowerBounds} revived the monotone factorisation perspective to formalise permutation oracles, while \citet{Adhikari2024RandomSamplingOfPermutationsThroughQuantumCircuits} adapted the classical Steinhaus–Johnson–Trotter algorithm~\cite{Steinhaus1979OneHundredProblemsInElementaryMathematics,Johnson1963GenerationOfPermutationsByAdjacentTransposition,Trotter1962AlgorithmPerm} for random sampling in quantum circuits, reversing the direction of adjacent swaps used in earlier works.
\citet{Bai2025AQuantumSpeedupAlgorithmForTSPBasedOnQuantumDynamicProgrammingWithVeryFewQubits} proposed an iterative insertion-based strategy which also relies on the monotone factorisation.

\begin{figure*}
    \begin{minipage}[t]{0.47\linewidth}
    \begin{algorithm}[H]\label{algorithm:FisherYates}
        \caption{\texttt{Fisher\_Yates}($a$)}
        \For{$i \leftarrow n - 1$ \KwTo $1$}{
            $j \leftarrow \text{random integer such that } 0 \leq j \leq i$\;
            exchange $a[j] \text{ and } a[i]$\;
        }
    \end{algorithm}
    \end{minipage}\hfill
    \begin{minipage}[t]{0.47\linewidth}
    \begin{algorithm}[H]\label{algorithm:ReverseFisherYates}
        \caption{\texttt{Reversed\_Fisher\_Yates}($a$)}
        \For{$i \leftarrow 1$ \KwTo $n - 1$}{
            $j \leftarrow \text{random integer such that } 0 \leq j \leq i$\;
            exchange $a[j] \text{ and } a[i]$\;
        }
    \end{algorithm}
    \end{minipage}
    \caption{\label{figure:FisherYatesShuffle}The Fisher-Yates shuffle as proposed by \citet{Durstenfeld1964AlgorithmRandomPermutation} and it's reversed version.}
\end{figure*}

Despite the proliferation of constructions, many recent approaches rediscover a small set of core techniques, often rooted in classical combinatorics.
This redundancy suggests a need for conceptual unification and principled design strategies.
In this work, we address this gap by presenting a streamlined derivation of quantum circuits based on the classical Fisher-Yates shuffle~\cite{Fisher1948StatisticalTablesForBiologicalAgriculturalAndMedicalResearch}. This approach not only recovers the original circuit of~\citet{Barenco1995ElementaryGatesForQuantumComputation} as a special case, but also provides a clear path from classical random sampling to coherent quantum state preparation.

By leveraging the Fisher–Yates algorithm and coupling it with efficient quantum subroutines for uniform superposition preparation~\cite{Shukla2024AnEfficientQuantumAlgorithmForPreparationOfUniformQuantumSuperpositionStates}, we construct quantum circuits that are asymptotically optimal in terms of qubit count, gate complexity, and depth among existing algorithms.
We provide a publicly available Qiskit implementation of these circuits is available under \cite{Binkowski2025QuantumFisherYates}.
In doing so, we aim to provide a unifying and practically efficient framework for permutation-based quantum algorithms across a wide range of applications.

\section{\label{section:TheObjective}The objective}

Let us first precisely state the objective we wish to tackle with our circuit construction.
For a given number $n \in \N$, we consider $n$ multi-qubit registers, each comprising $\lceil\log_{2}(n)\rceil$ qubits.
A single permutation $\sigma \in \symgroup_{n}$, expressed in inverted ``one-line notation''/``word representation'' $\sigma^{-1}(0)\, \sigma^{-1}(1)\, \cdots\, \sigma^{-1}(n - 1)$, is represented by the product state
\begin{align}
    \ket{\sigma}_{\permutation} \coloneqq \bigotimes_{k = 0}^{n - 1} \ket{\sigma^{-1}(k)}_{\permutation : k},
\end{align}
where the state $\ket{\sigma(k)}_{\permutation : k}$ is the binary encoding of the integer $\sigma(k)$ in the computational basis and belongs to the $k$-th $\lceil \log_{2}(n)\rceil$-qubit register, respectively.

Our first goal is to construct a unitary $A$ of length polynomial in $n$, acting on the $n \lceil\log_{2}(n)\rceil$-qubit permutation register $\hil_{\permutation}$ and an ancilla register $\hil_{\ancilla}$ so that
\begin{align}\label{equation:DisentanglingQuantumFisherYatesStateGeneration}
    A \ket{\bm{0}}_{\permutation} \otimes \ket{0}_{\ancilla} = \Bigg(\frac{1}{\sqrt{n!}} \sum_{\sigma \scriptin \symgroup_{n}} \ket{\sigma}_{\permutation}\Bigg) \otimes \ket{0}_{\ancilla}.
\end{align}
The main characteristic of $A$, apart from creating the uniform superposition of all permutations of $n$ elements in the register $\hil_{\permutation}$, is to leave $\hil_{\permutation}$ and $\hil_{\ancilla}$ disentangled after the computation.
We can reduce the circuit length at the expense of additional ancilla qubits and loosing the property of leaving $\hil_{\permutation}$ and $\hil_{\ancilla}$ disentangled.
That is, we also consider the less restrictive task of constructing a unitary $\tilde{A}$ acting on $\hil_{\permutation} \otimes \hil_{\ancilla}$ so that
\begin{align}\label{equation:EntanglingQuantumFisherYatesStateGeneration}
    \tilde{A} \ket{\bm{0}}_{\permutation} \otimes \ket{0}_{\ancilla} = \frac{1}{\sqrt{n!}} \sum_{\sigma \scriptin \symgroup_{n}} \big(\ket{\sigma}_{\permutation} \otimes \ket{\phi_{\sigma}}_{\ancilla}\big)
\end{align}
with some arbitrary ancilla states $\ket{\phi_{\sigma}}_{\ancilla}$.

Our second goal is to create a uniform superposition of all permutations of an $n$-fold composite register $\hil_{\data}$ whose $n$ subregisters $\hil_{\data : k}$, $0 \leq k \leq n - 1$, have $m$ qubits, respectively.
In doing so, we record the respectively applied permutations in a second register.\footnote{
    With this construction, we consciously avoid attempting to place e.g.\ the graph isomorphism problem in BQP (see \cite{Aharonov2003AdiabaticQuantumStateGenerationAndStatisticalZeroKnowledge}).
}
That is, we consider the task of constructing a shuffling unitary $B$, acting on $\hil_{\data} \otimes \hil_{\permutation} \otimes \hil_{\ancilla}$ so that
\begin{align}\label{equation:DisentanglingQuantumFisherYatesShuffle}
\begin{split}
    &B \bigotimes_{k = 0}^{n - 1} \ket{\psi_{k}}_{\data : k} \otimes \ket{\bm{0}}_{\permutation} \otimes \ket{0}_{\ancilla} \\
    &\quad\ = \Bigg(\frac{1}{\sqrt{n!}} \sum_{\sigma \scriptin \symgroup_{n}} \bigotimes_{k = 0}^{n - 1} \ket{\psi_{\sigma^{-1}(k)}}_{\data : k} \otimes \ket{\sigma}_{\permutation}\Bigg) \otimes \ket{0}_{\ancilla}.
\end{split}
\end{align}

Again, at the expense of additional ancilla qubits and remaining entanglement between the data, permutation, and ancilla register, we may reduce the circuit length.
The corresponding task is to construct a shuffle unitary $\tilde{B}$ acting on 
$\hil_{\data} \otimes \hil_{\permutation} \otimes \hil_{\ancilla}$ so that
\begin{align}\label{equation:EntanglingQuantumFisherYatesShuffle}
\begin{split}
    &\tilde{B} \bigotimes_{k = 0}^{n - 1} \ket{\psi_{k}}_{\data : k} \otimes \ket{\bm{0}}_{\permutation} \otimes \ket{0}_{\ancilla}\\
    &\quad\ = \frac{1}{\sqrt{n!}} \sum_{\sigma \scriptin \symgroup_{n}} \Bigg(\bigotimes_{k = 0}^{n - 1} \ket{\psi_{\sigma^{-1}(k)}}_{\data : k} \otimes \ket{\sigma}_{\permutation} \otimes \ket{\phi_{\sigma}}_{\ancilla}\Bigg)
\end{split}
\end{align}
with some arbitrary ancilla states $\ket{\phi_{\sigma}}_{\ancilla}$.
Note that in this second version, we may also omit the permutation register $\hil_{\permutation}$ entirely since the information about which permutation has been applied is uniquely encoded in the ancilla register $\hil_{\ancilla}$.
We state this task explicitly as constructing a light shuffle unitary $C$ acting on $\hil_{\data} \otimes \hil_{\ancilla}$ so that
\begin{align}\label{equation:EntanglingQuantumFisherYatesLightShuffle}
\begin{split}
    &C \bigotimes_{k = 0}^{n - 1} \ket{\psi_{k}}_{\data : k} \otimes \ket{0}_{\ancilla}\\
    &\quad\ = \frac{1}{\sqrt{n!}} \sum_{\sigma \scriptin \symgroup_{n}} \Bigg(\bigotimes_{k = 0}^{n - 1} \ket{\psi_{\sigma^{-1}(k)}}_{\data : k} \otimes \ket{\phi_{\sigma}}_{\ancilla}\Bigg).
\end{split}
\end{align}

\begin{figure*}[!th]
    \begin{minipage}[t]{0.47\linewidth}
        \begin{algorithm}[H]\label{algorithm:DisentanglingStatePreparation}
            \caption{$\vphantom{\tilde{A}}A(\hil_{\permutation}, \hil_{\ancilla}, n)$}
            \For{$k_{\vphantom{\ancilla}} \leftarrow 1$ \KwTo $n - 1$}{
                apply $\x_{0 \to k}$ to the subregister $\hil_{\permutation : k}$\;
            }
            \For{$i \leftarrow 1$ \KwTo $n - 1$}{
                apply $U_{i}$ on the register $\hil_{\ancilla}$\;
                \For{$j \leftarrow 0$ \KwTo $i - 1$}{
                    control-swap $\hil_{\permutation : j}$ and $\hil_{\permutation : i}$ via $C_{j}^{\ancilla}(\swap_{j, i})$\;
                }
                \For{$j \leftarrow 1$ \KwTo $i$}{
                    apply $C_{i}^{\permutation : j}(\x_{j \to 0})$ to the register $\hil_{\ancilla}$\;
                }
            }
        \end{algorithm}
    \end{minipage}\hfill
    \begin{minipage}[t]{0.47\linewidth}
        \begin{algorithm}[H]\label{algorithm:EntanglingStatePreparation}
            \caption{$\tilde{A}(\hil_{\permutation}, \hil_{\ancilla}, n)$}
            \tcc{$\hil_{\ancilla}$ is assumed to be composed of subregisters $\hil_{\ancilla : i}$, $1 \leq i \leq n - 1$ with $\numbits{i}$ qubits, respectively}
            \For{$k \leftarrow 1$ \KwTo $n - 1$}{
                apply $\x_{0 \to k}$ on the register $\hil_{\permutation : k}$\;
            }
            \For{$i \leftarrow 1$ \KwTo $n - 1$}{
                apply $U_{i}$ on the register $\hil_{\ancilla : i}$\;
                \For{$j \leftarrow 0$ \KwTo $i - 1$}{
                    control-swap $\hil_{\permutation : j}$ and $\hil_{\permutation : i}$ via $C_{j}^{\ancilla : i}(\swap_{j, i})$\;
                }
            }
        \end{algorithm}
    \end{minipage}
\caption{\label{figure:QuantumFisherYatesStateGeneration}
    Quantum circuit constructions for the disentangling and for the entangling state preparation.
    For a given $n \in \N$, both algorithms input an $n \lceil \log_{2}(n)\rceil$-qubit register $\hil_{\permutation}$, initialised to all-zero.
    \autoref{algorithm:DisentanglingStatePreparation} operates on an $\lceil\log_{2}(n)\rceil$-qubit ancilla register $\hil_{\ancilla}$, while \autoref{algorithm:EntanglingStatePreparation} requires $\sum_{i = 1}^{n - 1} \numbits{i}$ ancilla qubits since, in each iteration $1 \leq i \leq n - 1$, it operates on a different $\numbits{i}$-qubit subregister $\hil_{\ancilla : i}$.
    Apart from the uncomputation of the ancilla register (lines 9-11 in \autoref{algorithm:DisentanglingStatePreparation}), both algorithms operate in the same way:
    They first initialise the register $\hil_{\permutation}$ in the identity's (inverted) word representation, then iteratively bring the (respective subregister of the) ancilla register in a uniform superposition of the first $i + 1$ computational basis states, and, for all $0 \leq j \leq i$, swap the subregisters $\hil_{\permutation : j}$ and $\hil_{\permutation : i}$ controlled on whether the ancilla register is in the state $\ket{j}_{\ancilla}$.
}
\end{figure*}

\section{\label{section:QuantumFisherYatesShuffle}Quantum Fisher-Yates shuffle}

The classical Fisher-Yates shuffle~\cite{Fisher1948StatisticalTablesForBiologicalAgriculturalAndMedicalResearch}, assuming access to a true random number generator, shuffles an input array of length $n$ with a random permutation uniformly drawn from $\symgroup_{n}$.
By simply inputting an array with $a[k] = k$, i.e.\ the identity's (inverted) word representation, we may therefore use the algorithm to directly generate the inverted word representation of a random permutation of $n$ elements drawn uniformly from $\symgroup_{n}$.
The original shuffle operates in time $\bigo(n^{2})$.
Subsequently, the algorithm has been improved by \citet{Durstenfeld1964AlgorithmRandomPermutation} to operate in time $\bigo(n)$, which clearly constitutes the optimal runtime attainable for this task involving acting on $n$ elements.
Durstenfeld's version simply iterates through all array indices $i$ except the first one in descending order, samples for each $i$ an exchange index $j$ uniformly at random between $0$ and $i$, and exchanges the array elements at position $i$ and $j$ (see \autoref{algorithm:FisherYates} for the pseudo code).

An equally valid approach is to iterate through the array index $i$ in ascending order, but to leave the rest of the algorithm at it is (see \autoref{algorithm:ReverseFisherYates} for corresponding pseudo code).
We use this reversed version as the basis for a quantised Fisher-Yates shuffle as it has the conceptual advantage that when terminating the loop prematurely, say after iteration $i$, the first $i + 1$ elements are permuted randomly with a permutation uniformly drawn from $\symgroup_{i + 1}$\footnote{
    The reversed Fisher-Yates shuffle with it's recursive structure is conceptually similar to the Steinhaus-Johnson-Trotter algorithm used for successive generation of all $n$-element permutations.
}
It is fundamentally based on the observation that when building all the compositions of an element $\sigma \in \symgroup_{i}$ with one of the transpositions $\tau_{j, i} \coloneqq (j\, i)$, $0 \leq j \leq i$, which exchange the numbers $j$ and $i$, one obtains the entire $\symgroup_{i + 1}$.
Note that this an equivalent, recursive reformulation of the uniqueness of the monotone factorisation and also formalised in the subsequent lemma.

\begin{lemma}\label{lemma:RecursiveGenerationOfSymmetricGroup}
    Let $i \in \N$.
    For a permutation $\sigma \in \symgroup_{i}$, let $\tilde{\sigma} \in \symgroup_{i + 1}$ be defined via $\tilde{\sigma}(k) \coloneqq \sigma(k)$ for $0 \leq k \leq i - 1$ and $\tilde{\sigma}(i) \coloneqq i$.
    Then there is a bijection between $\{0, \ldots, i\} \times \symgroup_{i}$ and $\symgroup_{i + 1}$ via $(j, \sigma) \mapsto \tau_{j, i} \tilde{\sigma}$.
\end{lemma}

\begin{proof}
    Since both sets have a cardinality of $(i + 1)!$, it suffices to establish injectivity of the proposed mapping.
    To this end, let $(j, \sigma), (k, \pi) \in \{0, \ldots, i\} \times S_{i}$ such that $\tau_{j, i} \tilde{\sigma} = \tau_{k, i} \tilde{\pi}$.
    Since both $ \tilde{\sigma}$ and $\tilde{\pi}$ act trivially on the element $i$, we immediately obtain
    \begin{align*}
        j = \tau_{j, i}(i) = \big(\tau_{j, i} \tilde{\sigma}\big)(i) = \big(\tau_{k, i} \tilde{\pi}\big)(i) = \tau_{k, i}(i) = k.
    \end{align*}
    Multiplying both sides of $\tau_{j, i} \tilde{\sigma} = \tau_{j, i} \tilde{\pi}$ with $\tau_{j, i}$ from the left then establishes $\tilde{\sigma} = \tilde{\pi}$, i.e.\ $\sigma = \pi$.
\end{proof}

We quantise \autoref{algorithm:ReverseFisherYates} in a straight-forward way by keeping the overall loop, but replacing lines 2 and 3 with creating a uniform superposition of all states $\ket{j}_{\ancilla}$ with $0 \leq j \leq i$ in the ancilla register and subsequently controlling each of the swaps of the $j$-th and $i$-th subregister of the permutation register (and the data register in case of shuffling) on the ancilla register being in the state $\ket{j}_{\ancilla}$, respectively.
Prior to this, we initialise the permutation register in the identity's (inverted) word representation to generate the desired uniform superposition of all permutations' inverted word representations.

\begin{figure*}[!th]
    \begin{minipage}[t]{0.47\linewidth}
        \begin{algorithm}[H]\label{algorithm:DisentanglingShuffle}
            \caption{$\vphantom{\tilde{B}}B(\hil_{\data}, \hil_{\permutation}, \hil_{\ancilla}, n)$}
            \For{$k_{\vphantom{\ancilla}} \leftarrow 1$ \KwTo $n - 1$}{
                apply $\x_{0 \to k}$ to the subregister $\hil_{\permutation : k}$\;
            }
            \For{$i \leftarrow 1$ \KwTo $n - 1$}{
                apply $U_{i}$ on the register $\hil_{\ancilla}$\;
                \For{$j \leftarrow 0$ \KwTo $i - 1$}{
                    control-swap $\hil_{\data : j}$ and $\hil_{\data : i}$ via $C_{j}^{\ancilla}(\swap^{\data}_{j, i})$\;
                    control-swap $\hil_{\permutation : j}$ and $\hil_{\permutation : i}$ via $C_{j}^{\ancilla}(\swap^{\permutation}_{j, i})$\;
                }
                \For{$j \leftarrow 1$ \KwTo $i$}{
                    apply $C_{i}^{\permutation : j}(\x_{j \to 0})$ to the register $\hil_{\ancilla}$\;
                }
            }
        \end{algorithm}
    \end{minipage}\hfill
    \begin{minipage}[t]{0.47\linewidth}
        \begin{algorithm}[H]\label{algorithm:EntanglingShuffle}
            \caption{$\tilde{B}(\hil_{\data}, \hil_{\permutation}, \hil_{\ancilla}, n)$}
            \tcc{$\hil_{\ancilla}$ is assumed to be composed of subregisters $\hil_{\ancilla : i}$, $1 \leq i \leq n - 1$ with $\numbits{i}$ qubits, respectively}
            \For{$k \leftarrow 1$ \KwTo $n - 1$}{
                apply $\x_{0 \to k}$ on the register $\hil_{\permutation : k}$\;
            }
            \For{$i \leftarrow 1$ \KwTo $n - 1$}{
                apply $U_{i}$ on the register $\hil_{\ancilla : i}$\;
                \For{$j \leftarrow 0$ \KwTo $i - 1$}{
                    control-swap $\hil_{\data : j}$ and $\hil_{\data : i}$ via $C_{j}^{\ancilla : i}(\swap^{\data}_{j, i})$\;
                    control-swap $\hil_{\permutation : j}$ and $\hil_{\permutation : i}$ via $C_{j}^{\ancilla : i}(\swap^{\permutation}_{j, i})$\;
                }
            }
        \end{algorithm}
    \end{minipage}
\caption{\label{figure:QuantumFisherYatesShuffle}
    Quantum circuit constructions for the disentangling and for the entangling shuffle.
    For a given $m, n \in \N$, both algorithms input an $m n$-qubit data register $\hil_{\data}$ whose $n$ subregisters $\hil_{\data : k}$, $0 \leq k \leq n - 1$, are mutually disentangled and initialised arbitrarily.
    Additionally, they input an $n \lceil \log_{2}(n)\rceil$-qubit register $\hil_{\permutation}$, initialised to all-zero.
    Analogously to the state preparation routines, \autoref{algorithm:DisentanglingShuffle} operates on an $\lceil\log_{2}(n)\rceil$-qubit ancilla register $\hil_{\ancilla}$, while \autoref{algorithm:EntanglingShuffle} requires $\sum_{i = 1}^{n - 1} \numbits{i}$ ancilla qubits.
    Except for the uncomputation of the ancilla register (lines 10-12 in \autoref{algorithm:DisentanglingShuffle}), both algorithms follow the same steps:
    First, they initialize the register $\hil_{\permutation}$ in the identity's (inverted) word representation, then iteratively bring the (respective subregister of the) ancilla register in a uniform superposition of the first $i + 1$ computational basis states, and, for all $0 \leq j \leq i$, swap the subregisters $\hil_{\data : j}$ and $\hil_{\data : i}$ as well as $\hil_{\permutation : j}$ and $\hil_{\permutation : i}$ controlled on whether the ancilla register is in the state $\ket{j}_{\ancilla}$.
}
\end{figure*}

When addressing the tasks \eqref{equation:DisentanglingQuantumFisherYatesStateGeneration} and \eqref{equation:DisentanglingQuantumFisherYatesShuffle} of leaving the permutation register (and the data register in case of shuffling) and the ancilla register disentangled after the execution of the circuit, we uncompute the ancilla state after each iteration.
This can be done efficiently by controlling a cascade of \x-gates necessary to transform $\ket{j}_{\ancilla}$ back to $\ket{0}_{\ancilla}$ on the $j$-th subregister of the permutation register being in the state $\ket{i}_{\data}$.\footnote{
    Note that such an approach is not feasible without recording the permutation history in the permutation register, since the control on the state $\ket{i}_{\permutation: j}$ implicitly requires the knowledge that the $i$-th element of the identity's (inverted) word representation prior to permutation was $i$.
}
In contrast, when addressing the less restrictive tasks \eqref{equation:EntanglingQuantumFisherYatesStateGeneration} and \eqref{equation:EntanglingQuantumFisherYatesShuffle}, we simply leave the ancilla register entangled with the permutation register (and the data register in case of shuffling) and introduce for each iteration an additional ancilla register comprising $\numbits{i} \coloneqq \lfloor\log_{2}(i)\rfloor + 1$ qubits.

For the creation of the uniform superposition of all states $\ket{j}_{\ancilla}$, $0 \leq j \leq i$, we utilise a quantum circuit $U_{i}$, following a construction recently introduced by \citet{Shukla2024AnEfficientQuantumAlgorithmForPreparationOfUniformQuantumSuperpositionStates}.
$U_{i}$ operates solely on a $\numbits{i}$-qubit register without the need for additional ancilla qubits.
We detail their circuit construction in \autoref{section:ResourceAnalysis} when analysing it's gate and cycle count.

The just described steps give rise to our five algorithms addressing the five previously formulated in \autoref{section:TheObjective}:
a disentangling state preparation (\autoref{algorithm:DisentanglingStatePreparation}) implementing the unitary $A$ as defined in \eqref{equation:DisentanglingQuantumFisherYatesStateGeneration}, an entangling version of the state preparation (\autoref{algorithm:EntanglingStatePreparation}) implementing the unitary $\tilde{A}$ from \eqref{equation:EntanglingQuantumFisherYatesStateGeneration}, a disentangling shuffle (\autoref{algorithm:DisentanglingShuffle}) implementing the unitary $B$ as defined in \eqref{equation:DisentanglingQuantumFisherYatesShuffle}, an entangling version of the shuffle (\autoref{algorithm:EntanglingShuffle}) which implements the unitary $\tilde{B}$ from \eqref{equation:EntanglingQuantumFisherYatesShuffle}, and a light version of the entangling shuffle (\autoref{algorithm:EntanglingLightShuffle}) implementing $C$ from \eqref{equation:EntanglingQuantumFisherYatesLightShuffle}.
In the \hyperref[section:ProofOfCorrectness]{Appendix}, we supply a proof for the following theorem regarding the correctness of the disentangling shuffle.

\begin{theorem}\label{theorem:DisentanglingShuffleCorrectness}
    \autoref{algorithm:DisentanglingShuffle} implements the unitary $B$ as defined via \eqref{equation:DisentanglingQuantumFisherYatesShuffle}.
\end{theorem}

From this, we immediately obtain the state generation task \eqref{equation:DisentanglingQuantumFisherYatesStateGeneration} as a special case for $\hil_{\data} = \C$.
\begin{corollary}\label{corollary:DisentanglingStateGenerationCorrectness}
    \autoref{algorithm:DisentanglingStatePreparation} implements the unitary $A$ as defined via \eqref{equation:DisentanglingQuantumFisherYatesStateGeneration}.
\end{corollary}

Furthermore, by omitting the controlled bit flips $C_{i}^{\permutation : j}(\x_{j \to 0}^{\ancilla})$ and instead acting on a new subregister of the ancilla register in each iteration, we obtain \autoref{algorithm:EntanglingShuffle} from \autoref{algorithm:DisentanglingShuffle} and keep the validity of \eqref{equation:EntanglingQuantumFisherYatesShuffle}.
\begin{corollary}\label{corollary:EntanglingShuffleCorrectness}
    \autoref{algorithm:EntanglingShuffle} implements the unitary $\tilde{B}$ as defined via \eqref{equation:EntanglingQuantumFisherYatesShuffle}.
\end{corollary}

By simply omitting the permutation register $\hil_{\permutation}$ and all \swap-gates involving permutation subregisters, we further obtain \autoref{algorithm:EntanglingLightShuffle} from \autoref{algorithm:EntanglingShuffle}.
\begin{corollary}\label{corollary:EntanglingLightShuffleCorrectness}
    \autoref{algorithm:EntanglingLightShuffle} implements the unitary $C$ as defined via \eqref{equation:EntanglingQuantumFisherYatesLightShuffle}.
\end{corollary}

In the same way we derived \autoref{corollary:DisentanglingStateGenerationCorrectness} from \autoref{theorem:DisentanglingShuffleCorrectness}, we can proceed for the less restrictive tasks \eqref{equation:EntanglingQuantumFisherYatesShuffle} and \eqref{equation:EntanglingQuantumFisherYatesStateGeneration}.
\begin{corollary}\label{corollary:EntanglingStateGenerationCorrectness}
    \autoref{algorithm:EntanglingStatePreparation} implements the unitary $\tilde{A}$ as defined via \eqref{equation:EntanglingQuantumFisherYatesStateGeneration}.
\end{corollary}

We readily obtain the method introduced by \citet{Barenco1997StabilizationOfQuantumComputationsBySymmetrization} by enlarging the ancilla subregister $\hil_{\ancilla : i}$ in \autoref{algorithm:EntanglingLightShuffle} to contain $i$ qubits, respectively, by replacing $U_{i}$ with the uniform state preparation of the all-zero state and all $i$
computational basis states with Hamming weight one, and by concatenating this with applying the inverse ancilla state preparations and a subsequent measurement of the ancilla register.

\begin{figure}[!th]
\begin{algorithm}[H]\label{algorithm:EntanglingLightShuffle}
    \caption{$C(\hil_{\data}, \hil_{\ancilla}, n)$}
    \tcc{$\hil_{\ancilla}$ is assumed to be composed of subregisters $\hil_{\ancilla : i}$, $1 \leq i \leq n - 1$ with $\numbits{i}$ qubits, respectively}
    \For{$k \leftarrow 1$ \KwTo $n - 1$}{
        apply $\x_{0 \to k}$ on the register $\hil_{\permutation : k}$\;
    }
    \For{$i \leftarrow 1$ \KwTo $n - 1$}{
        apply $U_{i}$ on the register $\hil_{\ancilla : i}$\;
        \For{$j \leftarrow 0$ \KwTo $i - 1$}{
            control-swap $\hil_{\data : j}$ and $\hil_{\data : i}$ via $C_{j}^{\ancilla : i}(\swap^{\data}_{j, i})$\;
        }
    }
\end{algorithm}
\caption{\label{figure:QuantumFisherYatesLightShuffle}
    Quantum circuit constructions for the entangling light shuffle.
    For a given $m, n \in \N$, the algorithm inputs an $m n$-qubit data register $\hil_{\data}$ whose $n$ subregisters $\hil_{\data : k}$, $0 \leq k \leq n - 1$, are mutually disentangled and initialized arbitrarily, and an ancilla register $\hil_{\ancilla}$ comprising $\sum_{i = 1}^{n - 1} \numbits{i}$ qubits.
    The algorithm iteratively brings the respective ancilla subregister in a uniform superposition of the first $i + 1$ computational basis states, and, for all $0 \leq j \leq i$, swaps the subregisters $\hil_{\data : j}$ and $\hil_{\data : i}$ controlled on whether the ancilla register is in the state $\ket{j}_{\ancilla}$.
}
\end{figure}

\section{\label{section:ResourceAnalysis}Resource analysis}

Throughout this section, we assume $m \in \N$ and $n \in \N$ to be fixed, but arbitrary.
For a given number $k \in \N$, $\hamming(k)$ denotes it's Hamming weight, i.e.\ the number of ones in it's binary representation.
Furthermore, we denote with $\qubitc(W)$ the number of qubits on which a unitary $W$ acts non-trivially, with $\gatec(W)$ the gate count of $W$, and with $\cyclec(W)$ it's circuit depth/cycle count.

\subsection{\label{subsection:NumberOfQubits}Number of qubits}

First, the disentangling state preparation (\autoref{algorithm:DisentanglingStatePreparation}) acts on an $n \lceil \log_{2}(n)\rceil$-qubit permutation register $\hil_{\permutation}$ and uses an ancilla register $\hil_{\ancilla}$ for applying the swaps of subregisters of $\hil_{\permutation}$ in superposition.
Since the ancilla register's state is uncomputed after each iteration (lines 10-12), we can reuse it's qubits for the subsequent one.
Therefore, it's size is determined by how many qubits we need at most at once to represent any of the computational basis states on which we control the subregister swaps.
This maximum is readily attained in the last iteration where we have to store the number $n - 1$ in binary, yielding the required register size of $\numbits{n - 1} = \lceil\log_{2}(n)\rceil$.
In summary, we obtain
\begin{align}\label{equation:DisentanglingStatePreparationQubitCount}
    \qubitc(A) = n \lceil \log_{2}(n)\rceil + \lceil \log_{2}(n)\rceil = (n + 1) \lceil \log_{2}(n)\rceil.
\end{align}

Second, the entangling state preparation (\autoref{algorithm:EntanglingStatePreparation}) also operates on an $n \lceil \log_{2}(n)\rceil$-qubit permutation register $\hil_{\permutation}$, but requires a larger ancilla register $\hil_{\ancilla}$ since it does not uncompute the ancilla register after an iteration.
Therefore, the $i$-th iteration has to be carried out on a fresh set of qubits.
The largest number we have to represent during iteration $i$ is indeed the number $i$;
the corresponding subregister of $\hil_{\ancilla}$ has to contain at least $\numbits{i}$ qubits.
In summary, this yields a qubit count of
\begin{align}\label{equation:EntanglingStatePreparationQubitCount}
    \qubitc(\tilde{A}) = n \lceil \log_{2}(n)\rceil + \sum_{i = 1}^{n - 1} \numbits{i}.
\end{align}

Third, the disentangling shuffle (\autoref{algorithm:DisentanglingShuffle}) requires the same registers as the disentangling state preparation plus an additional data register to be shuffled.
Assuming $n$ data points and $m$ qubits per data point, this yields a total qubit count of
\begin{align}\label{equation:DisentanglingShuffleQubitCount}
    \qubitc(B) = \qubitc(A) + m n = (n + 1) \lceil \log_{2}(n)\rceil + m n.
\end{align}

Fourth, the qubit count for the entangling shuffle (\autoref{algorithm:EntanglingShuffle}) can be derived from the entangling state preparation analogously to the just derived qubit count \eqref{equation:DisentanglingShuffleQubitCount};
we obtain
\begin{align}\label{equation:EntanglingShuffleQubitCount}
    \qubitc(\tilde{B}) = \qubitc(\tilde{A})\hspace*{-1pt} +\hspace*{-1pt} m n = n \lceil \log_{2}(n)\rceil\hspace*{-1pt} +\hspace*{-3pt} \sum_{i = 1}^{n - 1} \numbits{i}\hspace*{-1pt} +\hspace*{-1pt} m n.
\end{align}

Fifth, in comparison to $\tilde{B}$, the light version of the entangling shuffle (\autoref{algorithm:EntanglingLightShuffle}) simply drops the permutation register, thus we arrive at a qubit count of
\begin{align}\label{equation:EntanglingLightShuffleQubitCount}
    \qubitc(C) = \qubitc(\tilde{B}) - n \lceil \log_{2}(n)\rceil = \sum_{i = 1}^{n - 1} \numbits{i} + m n.
\end{align}

\subsection{\label{subsection:GateCount}Gate count}

Next, we consider the gate counts for the five algorithms.
The initialisation of the permutation register, which is part of the first four methods, flips the qubits of the $k$-th subregister $\hil_{\permutation : k}$ to yield the binary representation of the number $k$, where $1 \leq k \leq n - 1$.
The number of \x-gates necessary to obtain $\ket{k}_{\permutation : k}$ from $\ket{0}_{\permutation : k}$ is simply the Hamming weight of $k$, $\hamming(k)$, i.e.\ the number of ones in $k$'s binary representation.
In each iteration $1 \leq i \leq n - 1$, we have to account for the cost of implementing the uniform state preparation $U_{i}$ of the first $i + 1$ computational basis states as well as of the implementation costs of the $i$ multi-controlled subregister swaps.
A full (multi-controlled) subregister swap is comprised of $\lceil \log_{2}(n)\rceil$ individual (multi-controlled) \swap-gates.
A \swap-gate, in turn, can be decomposed into three \cx-gates with alternating control and target qubit.
If the \swap-gate is controlled, it suffices to merely control the middle \cx-gate on the same qubits (compare \cite[Exercise 4.25]{Nielsen2010QuantumComputationAndQuantumInformation}).
Lastly, in order to disentangle the ancilla register $\hil_{\ancilla}$ after each iteration $i$, we apply, for all $1 \leq j \leq i$, a cascade of $\hamming(j)$ \x-gates transforming $\ket{j}_{\ancilla}$ to $\ket{0}_{\ancilla}$, but controlled on the subregister $\hil_{\permutation : j}$ being in the state $\ket{i}_{\permutation : i}$ which requires us to control each of the \x-gates on $\numbits{i}$ qubits.
In summary, we obtain a gate count of
\begin{align}\label{equation:DisentanglingStatePreparationGateCount}
\begin{split}
    \gatec(A) &= \sum_{i = 1}^{n - 1} \big(\hamming(i) \gatec(\x) + \gatec(U_{i}) + i \lceil \log_{2}(n)\rceil \times \\
    &\, (2 \gatec(\cx)\hspace*{-2pt} +\hspace*{-2pt} \gatec(\control^{\numbits{i}\hspace*{-1pt} +\hspace*{-1pt} 1}\x))\hspace*{-2pt} +\hspace*{-4pt} \sum_{j = 1}^{i}\hspace*{-3pt} \hamming(j) \gatec(\control^{\numbits{i}}\x)\big)
\end{split}
\end{align}
for the implementation of \autoref{algorithm:DisentanglingStatePreparation}.
For the entangling version of the state preparation, we simply omit the gate count for the uncomputations after each iteration, thereby obtaining
\begin{align}\label{equation:EntanglingStatePreparationGateCount}
\begin{split}
    \gatec(\tilde{A}) &= \sum_{i = 1}^{n - 1} \big(\hamming(i) \gatec(\x) + \gatec(U_{i}) \\
    &\quad + i \lceil \log_{2}(n)\rceil (2 \gatec(\cx) + \gatec(\control^{\numbits{i} + 1}\x))\big).
\end{split}
\end{align}

The gate count for the disentangling shuffle is the sum of the gate count \eqref{equation:DisentanglingStatePreparationGateCount} and the cumulative costs of implementing all the controlled subregister swaps of the data register $\hil_{\data}$, each of which is comprised of $m$ qubits.
Applying the same decomposition technique to those additional swap gates, we obtain
\begin{align}\label{equation:DisentanglingShuffleGateCount}
\begin{split}
    \gatec(B) &= \gatec(A) + \sum_{i = 1}^{n - 1} i m (2 \gatec(\cx) + \gatec(\control^{\numbits{i} + 1}\x)) \\
    &= \sum_{i = 1}^{n - 1} \big(\hamming(i) \gatec(\x) + \gatec(U_{i}) \\
    &\ + i (\lceil \log_{2}(n)\rceil\hspace*{-1pt} +\hspace*{-1pt} m) (2 \gatec(\cx)\hspace*{-1pt} +\hspace*{-1pt} \gatec(\control^{\numbits{i} + 1}\x)) \\
    &\ + \sum_{j = 1}^{i} \hamming(j) \gatec(\control^{\numbits{i}}\x)\big).
\end{split}
\end{align}

Analogously, we can derive the gate count of the entangling shuffle by adding the implementation cost of the data subregister swaps to the gate count \eqref{equation:EntanglingStatePreparationGateCount}, yielding a gate count of
\begin{align}\label{equation:EntanglingShuffleGateCount}
\begin{split}
    \gatec(\tilde{B}) &= \gatec(\tilde{A}) + \sum_{i = 1}^{n - 1} i m (2 \gatec(\cx) + \gatec(\control^{\numbits{i} + 1}\x)) \\
    &= \sum_{i = 1}^{n - 1} \big(\hamming(i) \gatec(\x) + \gatec(U_{i}) \\
    &\ + i (\lceil \log_{2}(n)\rceil\hspace*{-2pt} +\hspace*{-2pt} m) (2 \gatec(\cx)\hspace*{-2pt} +\hspace*{-2pt} \gatec(\control^{\numbits{i} + 1}\x))\big)
\end{split}
\end{align}
for the implementation of \autoref{algorithm:EntanglingShuffle}.
Subtracting the gate count for the initialisation of the permutation register as well as for the \swap-gates acting on permutation subregisters, we further arrive at the gate count for the entangling light shuffle:
\begin{align}\label{equation:EntanglingLightShuffleGateCount}
    \gatec(C) &= \sum_{i = 1}^{n - 1} \big(\gatec(U_{i})\hspace*{-2pt} +\hspace*{-2pt} i m (2 \gatec(\cx)\hspace*{-2pt} +\hspace*{-2pt} \gatec(\control^{\numbits{i} + 1}\x))\big).
\end{align}

\subsection{\label{subsection:CircuitDepth}Circuit depth}

Lastly, we consider the cycle counts of the five algorithms.
We make the following assumption which, however, only plays a role for the initialisation of the permutation register $\hil_{\permutation}$ in the first four methods:
Gates acting on disjoint qubits can be executed in parallel, constituting a single quantum cycle.
This assumption reduces the cycle count for the initialisation subroutine to the cycle count of a single \x-gate since all \x-gates act on distinct qubits of the permutation register.
For the remaining steps in \autoref{algorithm:DisentanglingStatePreparation}, we only give an upper bound on the cycle count in order to avoid distinguishing too many different cases (further parallelisation techniques are subject to an actual compiler/transpiler):
We simply assume that no gates of two iterations $1 \leq i_{1}, i_{2} \leq n - 1$ can be applied in parallel, hence the total cycle count is the sum of each iteration's cycle count.
Furthermore, we assume that there is no gate overlap between different subroutines within a single iteration.
Effectively, we may therefore replace all subroutines' gate counts in \eqref{equation:DisentanglingStatePreparationGateCount} with their respective cycle counts (except for the initialisation of the permutation register), yielding
\begin{align}\label{equation:DisentanglingStatePreparationCycleCount}
\begin{split}
    \cyclec(A) &= \cyclec(\x) + \sum_{i = 1}^{n - 1} \big(\cyclec(U_{i}) + i \lceil \log_{2}(n)\rceil (2 \cyclec(\cx) \\
    &\quad + \cyclec(\control^{\numbits{i} + 1}\x)) + \sum_{j = 1}^{i} \hamming(j) \cyclec(\control^{\numbits{i}}\x)\big).
\end{split}
\end{align}

The entangling state preparation's cycle count can be determined similarly to \eqref{equation:DisentanglingStatePreparationCycleCount}.
First, we simply omit the costs of implementing the uncomputation operations after each iteration.
Second, since each iteration acts on a different subregister of the ancilla register $\hil_{\ancilla}$, we can parallelise all the uniform state preparations $U_{i}$, $1 \leq i \leq n - 1$, on the ancilla subregisters.
Therefore, we only incur the maximum over all the $U_{i}$'s cycle counts rather than their sum.
In summary, we thus obtain
\begin{align}\label{equation:EntanglingStatePreparationCycleCount}
\begin{split}
    \cyclec(\tilde{A}) &= \cyclec(\x) + \max_{1 \leq i \leq n - 1}(\cyclec(U_{i})) \\
    &\quad + \sum_{i = 1}^{n - 1} i \lceil \log_{2}(n)\rceil (2 \cyclec(\cx) + \cyclec(\control^{\numbits{i} + 1}\x)).
\end{split}
\end{align}

The derivation of the cycle count for the disentangling shuffle follows the same assumptions as for the disentangling state preparation.
In addition, we assume the data subregister swaps to not overlap with any of the other steps.
This yields an overall cycle count of
\begin{align}\label{equation:DisentanglingShuffleCycleCount}
\begin{split}
    \cyclec(B) &= \cyclec(A) + \sum_{i = 1}^{n - 1} i m (2 \cyclec(\cx) + \cyclec(\control^{\numbits{i} + 1}\x)) \\
    &= \cyclec(\x) + \sum_{i = 1}^{n - 1} \big(\cyclec(U_{i}) + i (\lceil \log_{2}(n)\rceil + m) \times \\
    &\, (2 \cyclec(\cx)\hspace*{-2pt} +\hspace*{-2pt} \cyclec(\control^{\numbits{i}\hspace*{-1pt} +\hspace*{-1pt} 1}\x))\hspace*{-2pt} +\hspace*{-3pt} \sum_{j = 1}^{i}\hspace*{-2pt} \hamming(j) \cyclec(\control^{\numbits{i}}\x)\big)
\end{split}
\end{align}
for implementing \autoref{algorithm:DisentanglingShuffle}.
The costs for the entangling shuffle are analogously derived and yield
\begin{align}\label{equation:EntanglingShuffleCycleCount}
\begin{split}
    \cyclec(\tilde{B}) &= \cyclec(\tilde{A}) + \sum_{i = 1}^{n - 1} i m (2 \cyclec(\cx) + \cyclec(\control^{\numbits{i}}\x)) \\
    &= \cyclec(\x) + \max_{1 \leq i \leq n - 1}(\cyclec(U_{i})) \\
    &\ + \sum_{i = 1}^{n - 1} i (\lceil \log_{2}(n)\rceil\hspace*{-2pt} +\hspace*{-2pt} m) (2 \cyclec(\cx)\hspace*{-2pt} +\hspace*{-2pt} \cyclec(\control^{\numbits{i}}\x))
\end{split}
\end{align}
for the implementation of \autoref{algorithm:EntanglingShuffle}.
Omitting the initialisation of the permutation register as well as all \swap-gates acting on permutation subregisters, we obtain the cycle count for \autoref{algorithm:EntanglingLightShuffle}:
\begin{align}\label{equation:EntanglingLightShuffleCycleCount}
\begin{split}
    \cyclec(C) &= \max_{1 \leq i \leq n - 1}(\cyclec(U_{i})) \\
    &\ + \sum_{i = 1}^{n - 1} i m (2 \cyclec(\cx)\hspace*{-2pt} +\hspace*{-2pt} \cyclec(\control^{\numbits{i}}\x)).
\end{split}
\end{align}

\subsection{\label{subsection:AnalysisOfSubroutinesAndIndividualGates}Analysis of subroutines and individual gates}

Let us now inspect the gate and cycle count of the superposition creating circuit $U_{i}$ as proposed by~\cite{Shukla2024AnEfficientQuantumAlgorithmForPreparationOfUniformQuantumSuperpositionStates}.
In the following, let $i \in \N$ be arbitrary, but fixed.
If $i + 1 = 2^{r}$ for some $r \in \N$, then the generation of the uniform superposition of all computational basis $\ket{j}$, $0 \leq j \leq i$ degenerates to the simple application of \h-gates to each of the $\numbits{i}$ qubits.
In this case, we thus obtain $\gatec(U_{i}) = \numbits{i} \gatec(\h)$ and $\cyclec(U_{i}) = \cyclec(\h)$.
Otherwise, let $s_{0}, \ldots, s_{\hamming(i + 1)} \in \N$ denote the positions of the non-zero bits of the binary representation of $i + 1$, where we start counting from zero and $s_{0}$ corresponds to the position of the least significant bit.
Following the pseudo code in \cite[Algorithm 1]{Shukla2024AnEfficientQuantumAlgorithmForPreparationOfUniformQuantumSuperpositionStates}, we can easily infer a gate count of
\begin{align}\label{equation:ShuklaVedulaGateCount}
\begin{split}
    \gatec(U_{i}) &= s_{0} \gatec(\h) + \gatec(\ry) + (\hamming(i + 1) - 1) \gatec(\cry) \\
    &\quad + \sum_{k = 1}^{\hamming(i + 1)} \big(s_{k} - s_{k - 1}\big) \gatec(\ch).
\end{split}
\end{align}
From all operations within $U_{i}$, only the initial layer of Hadamard gates on all qubits before the least significant qubit may be applied in parallel, that is
\begin{align}\label{equation:ShuklaVedulaCycleCount}
\begin{split}
    \cyclec(U_{i}) &= \min(1, s_{0}) \cyclec(\h)\hspace*{-2pt} +\hspace*{-2pt} \cyclec(\ry)\hspace*{-2pt} +\hspace*{-2pt} (\hamming(i + 1) - 1) \times \\
    &\quad \cyclec(\cry) + \sum_{k = 1}^{\hamming(i + 1)} \big(s_{k} - s_{k - 1}\big) \cyclec(\ch).
\end{split}
\end{align}

This concludes the gate and cycle count for all ``high-level'' subroutines.
We expressed all their resource requirements in terms of implementation costs of \x-gates, \h-gates, and \ry-gates, their respective singly-controlled versions as well as multiply-controlled \x-gates.
The single-qubit gates $\x$, $\h$, and $\ry$ are typically native gates or can be synthesised from other native gates with small overhead \cite{Pino2021DemonstrationOfTheTrappedIonQuantumCCDComputerArchitecture,Graham2022MultiQubitEntanglementandAlgorithmsOnANeutralAtomQuantumComputer,IBM2023Heron}.
In practice, this means that $1 \leq \gatec(\x), \gatec(\h), \gatec(\ry) \leq 4$ (the cycle counts are identical to the gate counts) \cite{Pino2021DemonstrationOfTheTrappedIonQuantumCCDComputerArchitecture,Graham2022MultiQubitEntanglementandAlgorithmsOnANeutralAtomQuantumComputer,IBM2023Heron}.
Their controlled versions (with the partial exception of the $\cx$-gate) are usually not naively supported on any architecture and have to be synthesised from the available architecture-specific two-qubit gate and additional single-qubit gates.
The following bounds are realistic for non-photonic architectures: $1 \leq \gatec(\cx), \gatec(\ch), \gatec(\cry) \leq 39$ (again, the cycle counts match the gate counts) \cite{Pino2021DemonstrationOfTheTrappedIonQuantumCCDComputerArchitecture,Graham2022MultiQubitEntanglementandAlgorithmsOnANeutralAtomQuantumComputer,IBM2023Heron}.
The use of error-correcting codes will generally alter the set of natively available \emph{logical} gates, with the resulting set being highly dependent on the hardware and code employed \cite{Kubischta2023FamilyOfQuantumCodesWithExoticTransversalGates}.
In contrast, some error-mitigation techniques allow for utilising the full breadth of native physical gates \cite{Cai2023QuantumErrorMitigation, Koczor2021ExponentialErrorSuppressionForNearTermQuantumDevices, Huggins2021VirtualDistillationForQuantumErrorMitigation}.

Lastly, determining the gate and cycle count for $\control^{m}\x$-gates, $m \geq 2$, is more nuanced than the previous considerations.
In the special case $m = 2$, i.e.\ the case of the Toffoli gate, the usual ancilla-free (and also optimal~\cite{Shende2008OnTheCNOTCostOfTOFFOLIGates}) decomposition requires six $\cx$-gates and nine single-qubit gates (see e.g.\ \cite[Figure 4.9]{Nielsen2010QuantumComputationAndQuantumInformation}).
Without any additional ancilla qubits, the best known decompositions of a $\control^{m}\x$-gate requires $\bigo(2^{m})$ $\cx$-gates and single-qubit gates~\cite[Lemma 7.1]{Barenco1995ElementaryGatesForQuantumComputation}.
This can be drastically reduced once sufficiently many ancilla qubits are available:
For $m \geq 3$, a $\control^{m}\x$-gate can be decomposed into $2 (m - 1)$ Toffoli gates, utilising $m - 2$ ancilla qubits (see e.g.\ \cite[Figure 4.10]{Nielsen2010QuantumComputationAndQuantumInformation} for $m = 5$).
This construction assumes each of the ancilla qubits to be initialised in the $\ket{0}$ state, but also uncomputes the ancilla qubits.
In our case, however, we do not want to include more qubits than previously determined in \eqref{equation:DisentanglingStatePreparationQubitCount}--\eqref{equation:EntanglingLightShuffleQubitCount}.
Alternatively, we may always borrow qubits from a data/permutation subregister which is currently not swapped or does not control the uncomputation step for the disentangling versions.
Since these qubits are not guaranteed to be in the $\ket{0}$ state, we have to enlarge the decomposition to include $4 (m - 2)$ Toffoli gates for a toggle-based implementation and its uncomputation.
To each of these Toffoli gates we may now apply the previously discussed decomposition.
In summary, this yields $\gatec(\control^{m}\x) = 4 (m - 2) (6 \gatec(\cx) + 7 \gatec(\rz) + 2 \gatec(\h))$ and the corresponding cycle count $\cyclec(\control^{m}\x)$ is typically not substantially smaller.

\subsection{\label{subsection:AsymptoticResourceScaling}Asymptotic resource scaling}

From the previous quantitative resource analysis, we now infer the qualitative asymptotic scaling.
First, the scalings of the respective qubit counts are rather straightforward.
We readily infer that $\qubitc(A) \in \bigo(n \log(n))$ and that $\qubitc(B) \in \bigo((m + \log(n)) n)$.
For the other three algorithms, we first observe that for the ancilla register's size it holds that
\begin{align*}
    \sum_{i = 1}^{n} \numbits{i} &= \sum_{i = 1}^{n - 1} (\lfloor\log_{2}(i)\rfloor + 1) \leq n + \sum_{i = 1}^{n - 1} \log_{2}(i) \\
    &\leq n + \log_{2}\hspace*{-3pt}\Bigg(\prod_{i = 1}^{n} i\Bigg)\hspace*{-3pt} = n + \log_{2}(n!) \in \bigo(n \log(n))
\end{align*}
by Stirling's approximation.\footnote{
    In fact, one can further lower bound this quantity by the equally scaling sum of the first $n$ logarithms, thus arriving at the statement $\qubitc(\tilde{A}) \in \Theta(n \log(n))$, i.e.\ $\qubitc(\tilde{A}) \in \bigo(n \log(n))$ and $n \log(n) \in \bigo(\qubitc(\tilde{A}))$.
    Following the subsequent derivations thoroughly, one can actually replace ``$\bigo$'' with ``$\Theta$'' for the scaling of qubit, gate, and cycle count of all algorithms.
}
Thus, we conclude that $\qubitc(\tilde{A}) \in \bigo(n \log(n))$, $\qubitc(\tilde{B}) \in \bigo((m + \log(n)) n)$, and $\qubitc(C) \in \bigo((m + \log(n)) n)$.
Lastly, since the Shukla-Vedula construction for $U_{i}$ does not require any ancilla qubits, it holds that $\qubitc(U_{i}) \in \bigo(\log(i))$.

We continue with the asymptotic scaling of gate and cycle count.
As it turns out, both quantities scale identically for all presented algorithms; therefore, it is sufficient to focus on the gate count in the following.
On the most elementary level, we state that \x-gate, \h-gate, \ry-gate, and \rz-gate as well as their singly-controlled versions can be decomposed into a constant number of native gates.
This also implies that $\gatec(\control^{m}\x) \in \bigo(m)$.
Furthermore, a number's Hamming weight is trivially upper bounded by the number's binary length, hence $\hamming(i) \in \bigo(\log(i))$.
The gate counts for $A$ and $\tilde{A}$ are (asymptotically) clearly dominated by the costs of implementing $i$ permutation subregister swaps in each iteration $1 \leq i \leq n - 1$, scaling as $\bigo(i \log(n)^{2})$.
Executing the summation over all iterations establishes that $\gatec(A), \gatec(\tilde{A}) \in \bigo(n^{2} \log(n)^{2})$.
Similarly, for $B$ and $\tilde{B}$, the most costly components are the $i$ data and permutation register swaps per iteration, pushing their gate counts to lie in $\bigo((m + \log(n)) n^{2} \log(n))$.
With an analogous reasoning we also obtain that $\gatec(C) \in \bigo(m n^{2} \log(n))$.
Moreover, as already pointed out by \citet{Shukla2024AnEfficientQuantumAlgorithmForPreparationOfUniformQuantumSuperpositionStates}, the gate count for their state preparation circuit $U_{i}$ scales as $\bigo(\log(i))$.
The qubit and gate/cycle count complexities are again summarised in \autoref{table:AsymptoticScaling}.

\begin{table*}[!ht]
    \begin{minipage}{0.48\linewidth}
    \begin{tabular}{|c|c|c|}\hline
        algorithm   & qubit count (scaling)     & gate/cycle count (scaling) \\\hline
        $A$         & $\bigo(n \log(n))$        & $\bigo(n^{2} \log(n)^{2})$ \\
        $\tilde{A}$ & $\bigo(n \log(n))$        & $\bigo(n^{2} \log(n)^{2})$ \\
        $B$         & $\bigo((m + \log(n)) n)$  & $\bigo((m + \log(n)) n^{2} \log(n))$ \\
        $\tilde{B}$ & $\bigo((m + \log(n)) n)$  & $\bigo((m + \log(n)) n^{2} \log(n))$ \\
        $C$         & $\bigo((m + \log(n)) n)$  & $\bigo(m n^{2} \log(n))$ \\
        $U_{i}$     & $\bigo(\log(i))$          & $\bigo(\log(i))$ \\
        \hline
    \end{tabular}
    \caption{\label{table:AsymptoticScaling}
        Asymptotic scaling of qubit, gate, and cycle count for the state preparation and shuffling algorithms as well as for the state preparation circuit proposed by \citet{Shukla2024AnEfficientQuantumAlgorithmForPreparationOfUniformQuantumSuperpositionStates}.
        Here, $n$ denotes the number of elements/subregisters to permute and $m$ is the size of each data subregister.
    }
    \end{minipage}\hfill
    \begin{minipage}{0.48\linewidth}
    \begin{tabular}{|c|c|c|}\hline
        algorithm               & qubit count (scaling)     & gate/cycle count (scaling) \\\hline
        $A_{\text{oh}}$         & $\bigo(n \log(n))$        & $\bigo(n^{2} \log(n))$ \\
        $\tilde{A}_{\text{oh}}$ & $\bigo(n^{2})$        & $\bigo(n^{2} \log(n))$ \\
        $B_{\text{oh}}$         & $\bigo((m + \log(n)) n)$  & $\bigo((m + \log(n)) n^{2})$ \\
        $\tilde{B}_{\text{oh}}$ & $\bigo((m + n) n)$  & $\bigo((m + \log(n)) n^{2})$ \\
        $C_{\text{oh}}$         & $\bigo((m + n) n)$  & $\bigo(m n^{2})$ \\
        $V_{i}$                 & $\bigo(i)$          & $\bigo(i)$ \\
        \hline
    \end{tabular}
    \caption{\label{table:AsymptoticScalingOneHot}
        Asymptotic scaling of qubit, gate, and cycle count for the state preparation and shuffling algorithms with one-hot-encoded controls as well as for the state preparation circuit sketched by \citet{Barenco1997StabilizationOfQuantumComputationsBySymmetrization}.
        Here, $n$ is the number of elements/subregisters to permute and $m$ is the size of each data subregister.
    }
    \end{minipage}
\end{table*}

There is a generally applicable trade-off between the size of the ancilla register and the gate and cycle count of each of our methods:
Instead of using a $\lceil\log_{2}(n)\rceil$-qubit ancilla register for $A$ and $B$, we may introduce $n - 1$ ancilla qubits and prepare, in each iteration $1 \leq i \leq n - 1$, the uniform superposition of the all-zero state and $i$ computational basis states with Hamming weight one.
The respective state preparation circuit $V_{i}$ is sketched in \cite{Barenco1997StabilizationOfQuantumComputationsBySymmetrization} and requires $\bigo(i)$ gates and cycles.
This one-hot encoding of the control structures has the advantage that each of the \swap-gates only has to be controlled on a single qubit, avoiding the decomposition of multi-controlled gates with an $\bigo(\log(i))$ overhead.
The uncomputation of the ancilla qubits also becomes faster since, for each $1 \leq j \leq i$, only one multi-controlled \x-gate (instead of $\hamming(j)$ multi-controlled \x-gates) has to be applied to the ancilla register.
Overall, this technique eliminates one $\log(n)$ factor from the asymptotic scaling of gate and cycle count while introducing $\bigo(n)$ ancilla qubits.
At least for $A$ and $B$, the additional qubit requirements are still dominated by the $\bigo(n \log(n))$ qubit demand of the permutation register.
Applying the same technique to the entangling versions $\tilde{A}$, $\tilde{B}$, and $C$ also eliminates one $\log(n)$ factor from the gate/cycle demands, but increases the asymptotic qubit count:
In order to employ the state preparation $V_{i}$, each of the ancilla subregisters $\hil_{\ancilla : i}$ has to contain $i$ (instead of $\numbits{i}$) qubits, yielding an overall requirement of $\bigo(n^{2})$ ancilla qubits.
The resulting qubit, gate, and cycle complexities are summarised in \autoref{table:AsymptoticScalingOneHot}.

Finally, we compare the just derived qualitative qubit, gate, and cycle requirements with the ones of similar algorithms found in the literature.
Barenco~\textit{et~al.}'s projection into the symmetrised subspace is almost identical to $C_{\text{oh}}$, the entangling light shuffle with one-hot-encoded control, and indeed admits the same qubit count of $\bigo((m + n) n)$ and gate/cycle count of $\bigo(m n^{2})$.
In contrast, Chiew~\textit{et~al.}'s implementation as well as Adhikary's sampling algorithm are based on controlled sequences of adjacency transpositions and a binary encoding of the controls.
They are conceptually close to our algorithms $\tilde{A}$ and $C$, the entangling state preparation and the entangling light shuffle, with the exception that in each iteration $1 \leq i \leq n - 1$, they apply $\bigo(i^{2})$ controlled subregister swaps (instead of $\bigo(i)$ controlled subregister swaps).
This places their respective qubit requirements in $\bigo(n \log(n))$ and $\bigo((m + \log(n)) n)$, and their respective gate/cycle counts in $\bigo(n^{3} \log(n)^{2})$ and $\bigo(m n^{3} \log(n))$.\footnote{
    In their paper, \citet{Chiew2019GraphComparisonViaNonlinearQuantumSearch} derived a time complexity of $\bigo(n^{2} \log(n))$ for their state preparation routine which, unfortunately, is not correct as they did not take into account that each $i$ iteration applies $\bigo(i)$ many subregister swaps.
}
The method proposed by \citet{Bai2025AQuantumSpeedupAlgorithmForTSPBasedOnQuantumDynamicProgrammingWithVeryFewQubits} is conceptually close to our algorithm $A$, the disentangling state preparation.
Correspondingly, they also achieve a qubit requirement of $\bigo(n \log(n))$.
However, their gate and cycle requirements are of order $\bigo(n^{5 / 2} \log(n))$ and thus higher than for our \swap-based approach.\footnote{
    The extra factor $\log(n)$ which is not present in the runtime complexity derived by \citet{Bai2025AQuantumSpeedupAlgorithmForTSPBasedOnQuantumDynamicProgrammingWithVeryFewQubits} stems from decomposing the multi-controlled $\x$-gates in their algorithm.
}
Lastly, the TSP state preparation method introduced by \citet{Baertschi2020GroverMixersForQAOAShiftingComplexityFromMixerDesignToStatePreparation} is the only algorithm that does really arise as a special case of any of our algorithms.
This is due to the fact that they represent the permutations of $n$ elements as $n \times n$ permutation matrices whose entries correspond to single qubits.
In this lighter encoding of permutations ($n^{2}$ instead of $n \lceil\log_{2}(n)\rceil$ qubits), they can entirely avoid an ancilla register during the state preparation and have ``enough space'' to parallelise several controlled \swap-gates in each iteration.
In summary, they require $\bigo(n^{2})$ qubits, $\bigo(n^{3})$ gates, and $\bigo(n^{2})$ cycles.
That is, at the expense of a larger qubit demand, they are able to beat algorithm $A_{\text{oh}}$ in terms of cycle requirements, but also require more (unparallelised) gates.

\section{\label{section:Conclusion}Conclusion}

In this work, we presented a natural and efficient quantisation of the classical Fisher–Yates shuffle, leading to a family of quantum algorithms for preparing uniform superpositions over permutations of composite qubit registers. Our quantisation principle is straightforward: whenever a classical operation $D_{j}$ depends on a uniformly sampled integer $j \in J$, we replace it with a uniform superposition over all $j \in J$ encoded in an ancilla register and apply the quantised version of $D_{j}$, controlled on the ancilla being in state $\ket{j}$.

We proposed five algorithmic variants tailored to different objectives: constructing uniform superpositions over the inverted word representations of all permutations (either entangled with an ancilla or disentangled), generating superpositions of shuffled composite registers with the corresponding inverted word representations stored in an auxiliary register (again, with or without entanglement), and preparing uniform superpositions of shuffled registers entangled with an ancilla.
To support reproducibility and facilitate further research, we provide an open-source implementation of all algorithms using Qiskit.

Our resource analysis demonstrates that the quantum Fisher–Yates shuffle achieves an asymptotic optimum among known constructions targeting this problem class.
Specifically, it operates with only $\bigo(n \log(n))$ qubits and achieves a gate and cycle complexity of $\bigo(n^{2} \log(n))$.
In contrast, alternative approaches either require more qubits -- e.g., the method in~\cite{Baertschi2020GroverMixersForQAOAShiftingComplexityFromMixerDesignToStatePreparation} has a depth of only $\bigo(n^{2})$, but needs $\bigo(n^{2})$ qubits and $\bigo(n^{3})$ gates -- or incur higher gate and depth complexity when constrained to $\bigo(n \log(n))$ qubits, as in~\cite{Chiew2019GraphComparisonViaNonlinearQuantumSearch} and~\cite{Bai2025AQuantumSpeedupAlgorithmForTSPBasedOnQuantumDynamicProgrammingWithVeryFewQubits}.

These observations lead us to conjecture that the quantum Fisher–Yates shuffle may be optimal within a formal framework that remains to be precisely defined.
This intuition is bolstered by the classical Fisher–Yates shuffle’s optimal linear runtime for sampling uniform random permutations.
We strongly suspect that any alternative quantisation of a classically suboptimal method will not surpass the quantum Fisher–Yates construction.
Nevertheless, it remains possible that intrinsically quantum strategies, without classical analogues, might yield superior performance. We leave this as an open and intriguing direction for future research.

A natural extension of our work involves generalising the data register transformations beyond simple subregister swaps.
One could imagine using arbitrary unitary representations of the symmetric group acting on some Hilbert space $\hil_{\data}$, controlled by an ancilla $\hil_{\ancilla}$ that records the applied permutation in superposition.
While such a generalisation is straightforward in the entangled case, constructing disentangled versions is more subtle: these require the ancilla to be uncomputed after each iteration $1 \leq i \leq n - 1$, where transpositions $\tau_{j, i}$, $0 \leq j < i$, are applied in superposition.
If the action of a transposition cannot be unambiguously inferred from the resulting state of $\hil_{\data}$, uncomputation becomes impossible.
Determining which representations of the symmetric group admit this iterative uncomputation remains an open problem for future investigation.

\begin{acknowledgments}
    We thank Tobias J. Osborne for helpful discussions, pointing us to some relevant literature, and proof-reading parts of the manuscript.
    LB acknowledges financial support by the Quantum Valley Lower Saxony, the BMBF project QuBRA, and the Alexander von Humboldt Foundation.
    MS acknowledges financial support by the Alexander von Humboldt Foundation.

    \noindent\textbf{Data and code availability statement.}
    All code supporting the findings of this study is available at \cite{Binkowski2025QuantumFisherYates}.
\end{acknowledgments}

\bibliographystyle{apsrev4-2}
\bibliography{main.bib}

\appendix

\onecolumngrid

\section{\label{section:ProofOfCorrectness}Proof of Correctness}

\begin{proof}[Proof of \autoref{theorem:DisentanglingShuffleCorrectness}]
    Let $n \in \N$, then the initial state of the composite register $\hil_d \otimes \hil_p \otimes \hil_a$ after applying the cascade of \smash{$\x_{0 \to k}^{\permutation : k}$}-gates with $1 \leq k \leq n - 1$ (lines 1-3) is given by
    \begin{align*}
        \ket{\iota} \coloneqq \Bigg(\bigotimes_{k = 0}^{n - 1} \ket{\psi_{k}}_{\data : k}\Bigg) \otimes \Bigg(\bigotimes_{k = 0}^{n - 1} \ket{k}_{\permutation : k}\Bigg) \otimes \ket{0}_{\ancilla}.
    \end{align*}
    We will subsequently prove via an inductive argument that the state prepared before the $i$-th iteration of the main loop (lines 4-13) is given by $\ket{\xi_{i}}$ defined as
    \begin{align*}
        \ket{\xi_{i}} \coloneqq \frac{1}{\sqrt{i!}} \sum_{\sigma \scriptin \symgroup_{i}} \Bigg[ \underbrace{\Bigg(\bigotimes_{k = 0}^{i - 1} \ket{\psi_{\sigma^{-1}(k)}}_{\data : k}\Bigg) \otimes  \Bigg( \bigotimes_{k = i}^{n - 1} \ket{\psi_{k}}_{\data : k} \Bigg)}_{\in \hil_{\data}} \otimes \underbrace{\Bigg( \bigotimes_{k = 0}^{i - 1} \ket{\sigma^{-1}(i)} \Bigg) \otimes \Bigg(\bigotimes_{k = i}^{n - 1} \ket{k}_{\permutation : k} \Bigg)}_{\in \hil_{\permutation}} \otimes \ket{0}_{\ancilla} \Bigg].
    \end{align*}
    Clearly, for $i = 1$ one recovers the state $\ket{\iota}$ when considering that $\symgroup_{1}$ is the trivial group consisting of only the neutral element $\id$ with $\id = \id^{-1}$, whose one-line notation is given simply by $0$. This leaves open the task of proving that the $i$-th iteration maps $\ket{\xi_{i}}$ to $\ket{\xi_{i + 1}}$: Applying $U_{i}$ to $\ket{\xi_{i}}$ (line 5) yields
    \begin{align*}
        U_{i} \ket{\xi_{i}} = \frac{1}{\sqrt{(i + 1)!}} \sum_{\sigma \scriptin \symgroup_{i}} \sum_{j = 0}^{i} \Bigg[ \Bigg(\bigotimes_{k = 0}^{i - 1} \ket{\psi_{\sigma^{-1}(k)}}_{\data : k}\Bigg) \otimes \Bigg( \bigotimes_{k = i}^{n - 1} \ket{\psi_{k}}_{\data : k} \Bigg) \otimes \Bigg( \bigotimes_{k = 0}^{i - 1} \ket{\sigma^{-1}(k)}_{\permutation : k} \Bigg) \otimes \Bigg(\bigotimes_{k = i}^{n - 1} \ket{k}_{\permutation : k} \Bigg) \otimes \ket{j}_{\ancilla}  \Bigg]
    \end{align*}
    since $(i + 1)! = i! (i + 1)$. Subsequently, applying the product of all controlled register swaps $C^{\ancilla}_{j}(\swap^{\data}_{j, i})$ and $C^{\ancilla}_{j}(\swap^{\permutation}_{j, i})$ (lines 6-9), we obtain the state
    \begin{align*}
        \frac{1}{\sqrt{(i + 1)!}} \sum_{\sigma \scriptin \symgroup_{i}} \sum_{j = 0}^{i}
        \Bigg[ & \Bigg( \bigotimes_{k = 0}^{j - 1} \ket{\psi_{\sigma^{-1}(k)}}_{\data : k} \Bigg)
        \otimes \ket{\psi_{i}}_{\data : j}
        \otimes \Bigg( \bigotimes_{k = j + 1}^{i - 1} \ket{\psi_{\sigma^{-1}(k)}}_{\data : k} \Bigg)
        \otimes \ket{\psi_{\sigma^{-1}(j)}}_{\data : i}
        \otimes \Bigg( \bigotimes_{k = i + 1}^{n - 1} \ket{\psi_k}_{\data : k}\Bigg) \otimes \\
        \otimes & \phantom{\Bigg( \bigotimes_{k = 0}^{j - 1} \ket{\psi_{\sigma^{-1}(k)}}_{\data : k} \Bigg)}\mathllap{\Bigg( \bigotimes_{k = 0}^{j - 1} \ket{\sigma^{-1}(k)}_{\permutation : k} \Bigg)}
        \otimes \hphantom{\ket{\psi_{i}}_{\data : j}} \mathllap{\ket{i}_{\permutation : j} \;}
        \otimes \hphantom{\Bigg( \bigotimes_{k = j + 1}^{i - 1} \ket{\psi_{\sigma^{-1}(k)}}_{\data : k} \Bigg)}
        \mathllap{\Bigg( \bigotimes_{k = j + 1}^{i - 1} \ket{\sigma^{-1}(k)}_{\permutation : k} \Bigg)}
        \otimes \hphantom{\ket{\psi_{\sigma^{-1}(j)}}_{\data : i}} \mathllap{\ket{\sigma^{-1}(j)}_{\permutation : i}}
        \otimes \hphantom{\Bigg( \bigotimes_{k = i + 1}^{n - 1} \ket{\psi_k}_{\data : k}\Bigg)}
        \mathllap{\Bigg(\bigotimes_{k = i + 1}^{n - 1} \ket{k}_{\permutation : k} \Bigg) \hspace{0.25em}} \otimes \\
        \otimes & \hspace{0.5em} \ket{j}_{\ancilla} \Bigg].
    \end{align*}
    Importantly, we observe that the subregister $\hil_{\permutation : j}$ is in the state \smash{$\ket{i}_{\permutation : j}$} whenever the ancilla register is in the state~$\ket{j}_{\ancilla}$. Therefore, applying the cascade of controlled bit flips \smash{$C_{i}^{\permutation : j}(\x^{\ancilla}_{j \to 0})$} (lines 10-12) maps the ancilla state to $\ket{0}_{\ancilla}$ in each term, disentangling the ancilla register. This ultimately transforms this state into the posterior state
    \begin{align*}
        \ket{\xi_{i}'} = \frac{1}{\sqrt{(i + 1)!}} \sum_{\sigma \scriptin \symgroup_{i}} \sum_{j = 0}^{i}
        \Bigg[ & \Bigg( \bigotimes_{k = 0}^{j - 1} \ket{\psi_{\sigma^{-1}(k)}}_{\data : k} \Bigg) 
        \otimes \ket{\psi_{i}}_{\data : j}
        \otimes \Bigg( \bigotimes_{k = j + 1}^{i - 1} \ket{\psi_{\sigma^{-1}(k)}}_{\data : k} \Bigg)
        \otimes \ket{\psi_{\sigma^{-1}(j)}}_{\data : i}
        \otimes \Bigg( \bigotimes_{k = i + 1}^{n - 1} \ket{\psi_k}_{\data : k}\Bigg) \otimes \\
        \otimes &  \hphantom{\Bigg( \bigotimes_{k = 0}^{j - 1} \ket{\psi_{\sigma^{-1}(k)}}_{\data : k} \Bigg)}
        \mathllap{\Bigg( \bigotimes_{k = 0}^{j - 1} \ket{\sigma^{-1}(k)}_{\permutation : k} \Bigg)}
        \otimes \hphantom{\ket{\psi_{i}}_{\data : j}}
        \mathllap{\ket{i}_{\permutation : j} \;}
        \otimes \hphantom{\Bigg( \bigotimes_{k = j + 1}^{i - 1} \ket{\psi_{\sigma^{-1}(k)}}_{\data : k} \Bigg)} 
        \mathllap{\Bigg( \bigotimes_{k = j + 1}^{i - 1} \ket{\sigma^{-1}(k)}_{\permutation : k} \Bigg)}
        \otimes \hphantom{\ket{\psi_{\sigma^{-1}(j)}}_{\data : i}} 
        \mathllap{\ket{\sigma^{-1}(j)}_{\permutation : i}}
        \otimes \hphantom{\Bigg( \bigotimes_{k = i + 1}^{n - 1} \ket{\psi_k}_{\data : k}\Bigg)}
        \mathllap{\Bigg(\bigotimes_{k = i + 1}^{n - 1} \ket{k}_{\permutation : k} \Bigg) \Bigg]} \otimes \\
        \otimes & \ket{0}_{\ancilla}. \vphantom{\Bigg]}
    \end{align*}
    Now, note that for any given permutation $\sigma \in S_i$ and any $0 \leq j < i$ the equations
    \begin{align*}
        i & = \tilde{\sigma}^{-1}(i) = \tilde{\sigma}^{-1}\big( \tau_{j, i}^{-1}(j) \big) = (\tau_{j, i} \tilde{\sigma})^{-1} (j) \text{ and} \\
        \sigma^{-1}(j) & = \tilde{\sigma}^{-1}(j) = \tilde{\sigma}^{-1}\big( \tau_{j, i}^{-1}(i) \big) = (\tau_{j, i} \tilde{\sigma})^{-1}(i)
    \end{align*}
    hold, where $\tilde{\sigma}$ is defined as in \autoref{lemma:RecursiveGenerationOfSymmetricGroup}.
    Since $\tau_{j, i}$ leaves $0 \leq k < 0$ with $k \neq j$ invariant, we can additionally rewrite $\sigma^{-1}(k)$ as $\sigma^{-1}(k) = (\tau_{j, i} \tilde{\sigma})^{-1}(k)$. Hence, the posterior state $\ket{\xi_{i}'}$ is equivalently given by
    \begin{align*}
    \begin{aligned}
        \ket{\xi_{i}'} = \frac{1}{\sqrt{(i + 1)!}} \sum_{\sigma \in S_i} \sum_{j = 0}^{i} \Bigg[ & \Bigg( \bigotimes_{k = 0}^{n - 1} \ket{\psi_{(\tau_{j, i} \tilde{\sigma})^{-1}(k)}}_{\data : k} \Bigg) \otimes \Bigg( \bigotimes_{k = i + 1}^{n - 1} \ket{\psi_k}_{\data : k}\Bigg) \otimes \\
        \otimes & \Bigg( \bigotimes_{k = 0}^{n - 1} \ket{(\tau_{j, i} \tilde{\sigma})^{-1}(k)}_{\permutation : k} \Bigg) \otimes \Bigg(\bigotimes_{k = i + 1}^{n - 1} \ket{k}_{\permutation : k} \Bigg) \Bigg] \otimes \ket{0}_{\ancilla}.
    \end{aligned}
    \end{align*}
    The equality $\ket{\xi_{i + 1}} = \ket{\xi_{i}'}$ is thereby a consequence of \autoref{lemma:RecursiveGenerationOfSymmetricGroup}.
    Finally, \autoref{theorem:DisentanglingShuffleCorrectness} now follows from the observation that $\ket{\xi_{n}}$ is equal to $B \bigotimes_{k = 0}^{n - 1} \ket{\psi_{k}}_{\data : k} \otimes \ket{\bm{0}}_{\permutation} \otimes \ket{0}_{\ancilla}$ from~\eqref{equation:DisentanglingQuantumFisherYatesShuffle}.
\end{proof}

\twocolumngrid

\end{document}